\documentclass[11pt,letterpaper]{article}
\usepackage[paper=letterpaper,margin=1in]{geometry}

\usepackage{graphicx}
\usepackage{amssymb}
\usepackage{boxedminipage,afterpage}
\usepackage{amsfonts,  amsmath,  amsthm, amssymb}
\usepackage{caption,enumitem,array} 
\usepackage{graphicx,rotating}
\usepackage{url,color}
\usepackage{amscd,amsfonts,amsmath,amssymb,pifont}
\usepackage[colorlinks=true,citecolor=blue,urlcolor=black, linkcolor=blue]{hyperref}
\usepackage{setspace}

\usepackage{algorithm}
\usepackage{algpseudocode}

\usepackage{fullpage}

\title{Population stability: regulating size in the presence of an adversary
}

\date{}

 \author{
 {Shafi Goldwasser}\thanks{MIT.  Department of Computer Science. \url{shafi@csail.mit.edu}. Supported by NSF MACS CNS-1413920, DARPA IBM W911NF-15-C-0236, and a Simons Investigator award agreement dated 6-5-12.} \and
Rafail Ostrovsky\thanks{University of California, Los Angeles. Department of Computer Science and Mathematics.  
\url{rafail@cs.ucla.edu}.
Research supported in part by NSF grant 1619348, DARPA, US-Israel BSF grant 2012366, OKAWA Foundation
Research Award, IBM Faculty Research Award, Xerox Faculty Research Award, B. John Garrick Foundation Award,
Teradata Research Award, and Lockheed-Martin Corporation Research Award. The views expressed are those of the
authors and do not reflect position of the Department of Defense or the U.S. Government.}  \and
 Alessandra Scafuro\thanks{North Carolina State University.  Department of Computer Science.  \url{ascafur@ncsu.edu}.
}\and
Adam Sealfon \thanks{MIT. Department of Computer Science.  \url{asealfon@csail.mit.edu}.
Supported by a DOE CSGF fellowship, NSF MACS CNS-1413920, DARPA IBM W911NF-15-C-0236, and a Simons Investigator award agreement dated 6-5-12.
 }
 }

\newcommand{\interval}{T} %
\newcommand{\Tinner}{T_{\mathsf{inner}}}
\newcommand{\maxkill}{K}    
\newcommand{\nzero}{N}
\newcommand{\agentmemory}{M}

\newcommand{\N}{\nzero}
\newcommand{\advbound}{K}
\newcommand{\K}{\advbound}
\newcommand{\T}{\interval}

\newcommand{\al}{\alpha}
\newcommand{\halfal}{\frac{\alpha}{2}}
\newcommand{\oma}{(1-\al)}
\newcommand{\opa}{(1+\al)}
\newcommand{\omha}{(1-\halfal)}
\newcommand{\opha}{(1+\halfal)}
\newcommand{\opmha}{(1\pm\halfal)}

\newcommand{\round}{\mathsf{round}}
\newcommand{\mycolor}{\mathsf{color}} %
\newcommand{\torecruit}{\mathsf{to\_recruit}}
 
\newcommand{\amActive}{\mathsf{active}} 
\newcommand{\recruiting}{\mathsf{recruiting}}

\newcommand{\NbrIsER}{\mathsf{Nbr}.\SelfIsER} %
\newcommand{\NbrRec}{\mathsf{Nbr}.\recruiting}
\newcommand{\NbrColor}{\mathsf{Nbr}.\mycolor}
\newcommand{\NbrActive}{\mathsf{Nbr}.\amActive}
\newcommand{\SelfIsER}{\mathsf{inEvalPhase}} %

\newcommand{\MyStatus}{\mathsf{MyStatus}}
\newcommand{\NbrStatus}{\mathsf{Nbr}}
\newcommand{\Nbr}{\NbrStatus}

\newcommand{\Communicate}{Communicate}

\newcommand{\BiasedCoin}{TossBiasedCoin}
\newcommand{\ElectLeader}{DetermineIfLeader}
\newcommand{\CheckRoundConsistency}{CheckRoundConsistency}
\newcommand{\RecruitmentPhase}{RecruitmentPhase}
\newcommand{\EvaluationPhase}{EvaluationPhase}

\newcommand{\ExchangeMessages}{ExchangeMessages}

\newcommand{\rgets}{\stackrel{\$}{\gets}}
\newcommand{\binset}{\{0,1\}}
\newcommand{\splt}{\mathsf{split}}

\newcommand{\remove}[1]{}

\newtheorem{theorem}{Theorem}
\newtheorem{lemma}[theorem]{Lemma}

\newcommand{\negl}{\nu}
\newcommand{\polylog}{\mathrm{polylog}}
\newcommand{\poly}{\mathrm{poly}}

\begin{document}

\clearpage
\maketitle
\thispagestyle{empty}

\begin{abstract}

We introduce a new coordination problem in distributed computing that we call the {\bf population stability problem}.
A system of agents each with limited memory and communication, as well as the ability to 
replicate and self-destruct, is subjected to attacks by a worst-case adversary that can at a bounded rate 
(1) delete agents chosen arbitrarily and (2) insert additional agents with arbitrary initial state into the system. 
The goal is perpetually to maintain a population whose size is within a constant factor of the target size $N$.
The problem is inspired by the ability of complex biological systems composed of 
a multitude of memory-limited individual cells to maintain a stable population size in an adverse environment. 
Such biological mechanisms allow organisms to heal 
after trauma or to recover from excessive cell proliferation caused by inflammation, disease, or normal development.

We present a population stability protocol in a communication model that is a synchronous variant of the population model of Angluin et al. In each round, pairs of agents selected at random
meet and exchange messages, where at least a constant fraction of agents is matched in each round. 
Our protocol uses three-bit messages and $\omega(\log^2 \N)$ states per agent.
We emphasize that our protocol can handle an adversary that can both insert and delete agents, a setting
in which existing approximate counting techniques do not seem to apply.
The protocol relies on a novel coloring strategy in which 
the population size is encoded in the {\em variance} of the distribution of colors.  
Individual agents can locally obtain a weak estimate of the population
size by sampling from the distribution,
and make individual decisions that robustly maintain a stable global population size.

\end{abstract}
\vspace{50pt}
\setcounter{page}{0}
\pagebreak

\section{Introduction}

A single fertilized mouse egg and human egg 
develop into organisms with vastly different numbers of cells. 
How do cellular mechanisms
regulate the number of cells in complex biological systems?
In order to function in adverse environments, organisms must maintain %
stability and be able to recover from unplanned circumstances. For example, 
a lizard that loses its tail can grow a new one,
and internal organs 
require mechanisms to recover from cell loss caused by injury or from
execessive cell proliferation due to development or disease.
But how do individual cells know how to respond in order to reestablish the
desired population size?

Regulation of population size may be achieved through a combination of internal programs running within 
each cell and intercellular communication. 
One approach could be for individual cells to 
count the population using a distributed protocol. 
But an interesting question is how to control 
the population size if each cell lacks sufficient memory to count.
Understanding the mechanisms for regulating population size in an adversarial environment, in light of
memory constraints of individual agents, is a natural computational question. %

In this work, we study the problem of robustly maintaining a stable population size 
from the perspective of distributed computing.
We introduce a new coordination question that we call the {\bf population stability problem}.
Consider a population of agents with the ability to replicate and self-destruct.
How can such distributed systems detect and recover from adversarial deletions and 
insertions of agents so as to maintain the desired population size?

Our focus is on systems that consist of huge numbers of agents, where each
agent individually has very limited memory and connective capability and can
directly communicate with only a few other agents in the system.
We model communication using a synchronous variant of the population model 
of  Angluin et al.~\cite{AADFP06,AAE07}. In each round, a constant fraction of agents is matched
at random and can exchange messages, where the matched agents are chosen independently in each round.
Population size must be maintained in this setting in the presence of an adversary that observes
the entire state of the system and can continually delete or insert agents.
We describe the model in more detail below.

The population stability problem augments a growing body of work that uses the language and ideas of distributed
computing to model biological systems
consisting of a collection of resource-constrained components that collectively accomplish
complex tasks.
Naturally, we do not claim direct relevance of our results to biological systems due to potential modeling
differences.
Regardless, the population stability problem makes sense in
any system consisting of individual components with the ability to reproduce.

\subsection{Contributions}
\label{ssec:contributions}
Our main contributions are as follows.

\paragraph{A New Problem in Distributed Computing: The Population Stability Problem.}  We introduce  a new problem in distributed computing. A population of 
$\nzero$ memory-constrained agents (i.e.~processors with the ability to reproduce and self-destruct)  is 
subjected to adversarial attacks. Whereas many attacks can be envisioned, 
we consider a worst-case  adversary that can delete or insert agents  at a bounded rate.
The goal is to maintain a stable population size within a small multiplicative  factor of the original size $\nzero$.
This problem appears fundamentally different from the classical problems of distributed computing, 
such as consensus, leader election, majority, common coin flipping,  or computing
general functions of the joint state of the parties.

\paragraph{Models for Communication and the Adversary.}
The communication model we consider is a synchronous variant of the population model of ~\cite{AADFP06,AAE07}. That model  was designed to represent sensor networks consisting of very limited mobile agents with no control over their own movement and whose goal is to compute some function of their inputs or evaluate a property of the system. Whereas \cite{AADFP06} assumed 
that pairs of agents can communicate via pairwise interactions as scheduled by a uniformly random matching process,
we assume in addition that agents are synchronized and interact with one another in rounds.
Within each round, at least a $\gamma$ fraction of agents participates in pairwise interactions, again scheduled
uniformly at random. The agents 
additionally have the ability to self-destruct and to reproduce by producing a second identical copy of themselves.

It is clear that we cannot allow the adversary to delete most or all of the 
agents in a single round, since maintaining a stable population size in the presence of such an adversary is impossible. 
Consequently, we give the adversary a budget of $\maxkill$ alterations to perform in each round,
where an alteration consists of removing, inserting or modifying the memory of a single agent. 
We allow the adversary to observe the memory contents of \emph{every} agent before determining its alterations
for a round. 
Both $\gamma$ and $\maxkill$ are parameters of the model.
The model is described in detail in Section~\ref{model}.

\paragraph{Protocol for Population Stability.}
We present a  protocol with three-bit messages requiring $\polylog(N)$ states (i.e.~$\Theta(\log \log N)$ bits of memory)
per agent that tolerates $\maxkill = \N^{1/4 - \epsilon}$ worst-case insertions or deletions in each round, 
for any constant $\epsilon > 0$. Formally, our main theorem is the following.

\begin{theorem}
Let $\alpha,\gamma,\epsilon$ be positive constants, where $\gamma$ is a lower bound
on the fraction of agents that is matched in each round.
Then there exists a population stability 
protocol with three-bit messages and $\polylog(N)$ states per agent and
guaranteeing that if the adversary inserts or deletes at most $\maxkill = O(\N^{1/4-\epsilon})$
agents in each round, then with
all but negligible probability the population will remain between $(1-\alpha)\nzero$ and $(1+\alpha)\nzero$ 
for any polynomial number of rounds. 
\label{thm:mainthm}
\end{theorem}

\paragraph{New Techniques.}
The main idea employed in our construction is to \emph{color} the agents with the values $0$ and $1$
in such a way that information about the population size is encoded in the distribution of colors.
That is, given a set of agents with assigned colors in $\binset$, 
consider the distribution specified by choosing an 
agent at random and observing its color. We are able to assign colors to agents
in such a way that approximate population size is encoded in the \emph{variance} of this distribution. 
Subsequently, each individual agent locally computes a very weak estimator of whether the variance 
is too large or too small, and makes a local decision of whether to reproduce or self-destruct. 
Although each individual agent's estimate is noisy, we show that in the aggregate, the local decisions  
are globally able to maintain a stable population size even in the presence of a powerful adversary.

We discuss the model and results further in Section~\ref{ssec:discuss}. 
We provide a more in-depth overview of our techniques in Section~\ref{ssec:tech-overview} below.
The model is described in further detail in Section~\ref{model}. 
In Section~\ref{sec:protocol} we provide a full description of the protocol,
and in Section~\ref{sec:analysis} we present the analysis. 

\subsection{Discussion and extensions}
\label{ssec:discuss}
\paragraph{Challenges.}
The first obstacle we must confront in designing a protocol for population stability is the memory constraint. 
Each agent does not have sufficient memory to store a unique identifier or to count to the population target $\N$. 
Yet collectively, they must have a good approximation of 
$\nzero$ and must individually decide whether to replicate if the population is too low or to self-destruct if the population is too high.
Making their task even more challenging, these memory-constrained agents must 
correct deviations in the size of the population and make their decisions {\em while} the adversary is acting.

A second major challenge is that the adversary is very powerful. Although the number of agents
inserted or deleted in each round is bounded, the adversary can observe the complete state of the
system\footnote{That is, the adversary has the ability to read the memory contents of every agent.}
before deciding on its actions, and may insert agents of arbitrary initial state. 
These capabilities present difficulties for many standard techniques and abstractions
used in distributed computing. For instance, many protocols are based on leader election, where a single
leader processor is chosen to direct or facilitate the task at hand. However, since our adversary is
able to observe the state of every agent, the adversary can simply wait for a leader to be chosen and then
delete it. Furthermore, even without observing the internal memory of agents, 
the adversary could insert many additional agents that are all identical in state
to the leader. Indeed, since agents %
can be in one of only $\polylog(\N)$
distinct states and the adversary is able to insert $\K=\N^{1/4-\epsilon}\in\poly(\N)$ agents per round,
the adversary can even insert many  copies of \emph{every possible} agent type in each round. 
Consequently any approach that relies on the existence of agents of unique or extremely special state,
such as leader election, seems doomed to failure.
This appears to render ineffectual a large part of the distributed computing toolbox. 

\paragraph{Adversarial insertions.} Recall that we allow the adversary to insert new agents with arbitrary initial state.
Starting from that internal state, we assume that the inserted agents execute the same protocol as honest agents.
We could instead consider an even stronger adversary that inserts agents 
running arbitrary malicious protocols specifying their subsequent behavior. 
However, the population stability problem as described above is clearly impossible in the face of this stronger
adversary, since our model does not include the ability to destroy agents that do not cooperate. 
A malicious agent can simply ignore all interactions with other agents and replicate itself at every opportunity.
Such malicious agents would quickly replicate themseves out of control, rapidly exceeding the population target.

One may consider a different model that allows agents not only to self-destruct but also to remove other agents
it encounters.
In such a setting, 
our protocol can be extended to achieve population stability even if the adversary 
is allowed to insert agents that execute arbitrary malicious programs,
as long as there is a bound on how frequently malicious agents can replicate and 
an agent is able to detect when it encounters an agent whose program is different from its 
own. %
However, that setting is not the focus of this work.

\paragraph{Correction and detection.} The objective in the population stability problem is to maintain
a population size that is close to a target value $\N$, and to correct the population if it deviates too far from this target.
A closely related problem is that of simply \emph{detecting} whether the population has deviated too far from the target
or whether it has exceeded some threshold, objectives that seem very similar to the problem of
approximate counting~\cite{Mor78}. 
It would be natural to try to solve our problem by first detecting whether
the population is too large or too small and then correcting appropriately. 
With a weaker adversary that can only delete agents but not insert additional ones (and additionally is oblivious
to the internal states and coin flips of agents), this approach can be made to work using
approximate counting techniques. 
In the adversarial model considered here, 
constructing approximate counters and detecting changes in population size 
are interesting open questions.

\paragraph{Synchrony.} 
In this work we study a synchronous model, where all agents 
communicate and perform updates in rounds.
As has commonly been the case across distributed computing, it is natural to study a new problem first 
in a synchronous setting to distill key ideas and techniques before adding additional complications in an
asynchronous setting. 

We note, however, that synchronization is far from an unreasonable assumption in biological systems.
Indeed, many multicellular systems do achieve synchrony either through regular external stimuli such as sunlight 
or through chemical control mechanisms.
For instance, heart cells maintained in culture were able to achieve
a high degree of synchronization of their rhythmic contractions~\cite{JMT87cardiology}. Neuronal cells
too exhibit highly synchronized behavior. Even when grown in culture, specialized neuronal cells
show the capacity to synchronize the release of particular hormones at regular time intervals~\cite{EMCW92gnrh}.
Bacterial populations have also been shown to have the capability of producing coordinated oscillations
\cite{MO'S09molecularclock,DMTH10bacteria}.

Nonetheless, an extremely natural and interesting question is how to solve the population stability problem
in a setting without synchrony or with only partial synchrony. For instance, one could consider a setting where
agents have clocks that have bounded drift relative to one another. Related to this is the typical 
random scheduler setting of population protocols~\cite{AAE07}, in which a single pair of agents at a time is chosen to interact
and update state. By a concentration argument, this process allows agents to maintain clocks that do not drift
too quickly relative to one another. 

While the construction in this paper requires synchrony, there are some known techniques in the 
population protocol literature for maintaining approximate synchronization
in a non-synchronous setting; see, for instance, the recent work of~\cite{AAG18}. 
A natural extension is to show whether our techniques can be combined with synchronizers to achieve
population stability in such settings.

\paragraph{Alternate communication models.}
Another very interesting question is to explore 
the population stability problem under a different communication model. 
In this paper, pairs of agents that communicate are chosen independently at random in each round.
Alternatively, one could consider settings in which the neighbors
of an agent are consistent over time, perhaps reflecting underlying geometric constraints.

One natural approach is to use a fixed sparse communication graph, for instance an expander. However,
modeling problems arise in determining how connectivity changes upon agent replication, insertion, 
or deletion. In various settings along these lines, 
it is straightforward for the adversary to disconnect the communication
graph and consequently to violate population stability. 
An alternate approach could be to associate agents with points in $\mathbb{R}^d$, and to 
allow each agent to communicate with a small number of the nearest other agents.

\paragraph{Population stability in the high-memory setting.} 
We note that in the absence of memory constraints, there is a trivial protocol both for approximate counting and
for the population stability problem if the adversary can only delete and not insert new agents. 
Each agent simply flips $\N$ coins to generate
a unique identifier $\mathsf{id}\in\binset^\N$. For an interval of $O(\poly\log(\N))$ rounds each agent broadcasts
the set of identifiers it has received so far. 
With high probability, all identifiers are unique and agents receive the identifiers of every agent that was alive
throughout the interval, so each agent learns a close approximation of the population size and
can make a decision of whether to self-destruct, replicate, or neither. 
However, this protocol relies heavily on agents having very large memory, and does not 
yield an approach to solve the problem in the low-memory setting.

\paragraph{A note about success probability.}
Theorem~\ref{thm:mainthm} states that our protocol maintains a stable population for any polynomial rounds
with all but negligible probability. Recall that a function is \emph{negligible} if it tends to zero
faster than any polynomial. That is, a function $\negl(x)$ is negligible if for any $c\in \mathbb{N}$
there exists an $x_0$ such that $\left|\negl(x)\right| < 1/x^c$ for all sufficiently large $x \geq x_0$. 

Throughout this paper, we use the phrases ``with high probability'' and ``with overwhelming probability''
to mean with probability $1-\negl(\N)$ for some negligible function $\negl$.

\subsection{Technical overview}
\label{ssec:tech-overview}

\subsubsection{Preliminary attempts}

We first describe two preliminary attempts at protocols for the population stability problem which, while unsound,
will provide useful intuition toward the design of our actual protocol.

\paragraph{Attempt 1: non-interactive leader election}
As a first attempt in the low-memory setting,  consider the following approach, which is 
based on an idea from~\cite{AAEGR17}. 
Each agent flips a biased coin where $\Pr[c=1] = \frac{1}{\N}$, where outcome $1$ means that the agent is a leader. 
For $O(\polylog(\N))$
rounds, each agent sends any agent it encounters the bit $0$ if its coin was zero  and it has not received
the message $1$, and the bit $1$ if its coin was $1$ or it has received a $1$ from another agent.
In the absence of an adversary, this allows every agent to learn whether a $1$ was obtained in any
of the initial coin flips. The probability of this event differs noticeably depending on whether the population
is too small or too large. After repeating to amplify the signal, with high probability the agents can detect if
the population is too small or too large and can replicate or self-destruct accordingly.
A protocol of this form can be shown to work in the presence of an adversary that can only delete and 
not insert agents, and is additionally oblivious to the coin flips made by the agents.
However, in the adversary model considered here with insertion as well as deletion 
and full knowledge of the states of agents, the protocol will fail. The adversary can either insert an agent with coin value
$c=1$ in each phase, or else identify the agent or agents with coin value $1$ and selectively remove these agents.
Consequently the adversary can cause the population to grow or shrink arbitrarily. 

This attempt highlights a fundamental difficulty in designing protocols in our adversarial model. 
The protocol relied on a non-interactive strategy related to leader election, where 
the presence or absence of a leader could be used to infer the approximate size of the population.
However, as we have discussed above, 
the use of a special state with one or only a few agents of that state
(in this case agents with coin value $c=1$) provides the adversary with an easy avenue of attack, namely
the deletion of agents of that state or the insertion of many additional agents of that state. Consequently,
constructions of this flavor seem to have little promise in this adversarial setting. 

\paragraph{Attempt 2: independent coloring}
As a next attempt, consider the simple protocol in which each agent flips a fair coin, receiving 
at random a color $c\in\binset$. For each agent, compare the colors of the next two agents encountered.
If the colors are equal, then split, and otherwise self-destruct. Observe that if an agent encounters the same agent twice,
the colors must be the same,  while if an agent encounters two different agents, the colors are independently random.
Consequently if the population currently has size $m$, then
the probability of splitting is $\frac{1}{2} + \Theta(\frac{1}{m})$, which is slightly larger than the probability
$\frac{1}{2}-\Theta(\frac{1}{m})$ of self-destructing, and so this protocol would cause the population to
increase slightly over time. To compensate this, modify the protocol to 
split only with probability  $1-\Theta(\frac{1}{\N})$ if the colors are equal while still splitting
with probability $1$ if the colors are unequal.\footnote{Another perspective on why the protocol
without this step cannot maintain a stable population is
that the protocol run by each agent has no dependence on the population target $\N$. Consequently,
if agents are added or removed
(either by adversary, or even as a result of random drift)
then the protocol should behave as if those agents were there to
begin with and will not correct this deviation in the population.}

Now, the population will stay the same size in expectation if its current size $m$ is equal to the target $\N$,
will decrease in expectation if $m>\N$, and will increase in expectation if $m<\N$. 
Qualitatively, this is exactly the behavior that we want. 
However, tending in expectation to correct itself is insufficient for maintaining a stable population. 
In fact, despite a very weak bias to correct drifts in the population, the signal is overwhelmed by the
noise, and the size of the population under this protocol will behave very much like a random walk. 
Even in the absence of an adversary, this protocol will cause the population to drift extremely far from
its initial size. 

In some sense the protocol we have just outlined behaves even worse than the empty protocol,
in that it
fails to maintain a stable population when there is no adversary at all.
However, the protocol does have one intriguing feature. 
It entirely lacks any ``special'' agent types for the adversary to exploit. 
Consequently, if we could design a more sophisticated protocol along these lines
that could maintain a stable population in the absence of an adversary, we might hope
that it could do the same even in the presence of an adversary. 

\subsubsection{Overview of our protocol}
We now describe our actual protocol. 
At a very high level, the idea behind our protocol is as follows.
Through some \emph{coloring process} which we will discuss below, 
agents are colored with the colors $\binset$. 
After the agents are colored, agents run a step called the \emph{evaluation phase} in which agents make the decision
of whether to reproduce or self-destruct.  The coloring
process and evaluation phase are then repeated indefinitely. We will refer to each iteration of the coloring process
followed by the evaluation phase as an {\em epoch.}

During the evaluation phase, each 
agent that is matched with another agent in this round compares its own color with the color of its \emph{neighbor}
(i.e.~the agent to which it is matched). 
If the two agents have the same color, then the agent will replicate itself with some probability $p_\splt$.
If the two agents have different colors, then the agent will self-destruct.
Note that if the coloring process consisted of every agent tossing its own coin, then this would be essentially
the same as Attempt~2 above.  We will instead employ a more structured coloring process 
that results in agent colors that are not generated independently at random.

The coloring process will consist of two phases, a (noninteractive) \emph{leader selection phase} and a
\emph{recruitment phase}.
In the leader selection phase, $\Theta(1/\sqrt{\N})$ of the agents will become ``leaders''\footnote{That is, each agent will become a leader with probability $\Theta(1/\sqrt{\N})$.} and each leader will choose a random color in $\binset$.
In the recruitment phase, each leader will identify $\sqrt{\N}$ uncolored agents and color each of them  
with its own color. We note that the leader will not directly encounter each of these $\sqrt{\N}$ agents,
but will directly color some agents which in turn will color other agents.
To begin with, each leader activates the first inactive agent that it encounters, sharing its color with
the new agent. Each agent is subsequently responsible for recruiting $\sqrt{\N}/2$ inactive agents
in the same manner, forming a recruitment tree of depth $\frac{1}{2}\log \N$.
By delegating the coloring in this manner, the recruitment process can be performed in only $O(\log^2\N)$ 
rounds.\footnote{In order to achieve constant message size, our full protocol will be slightly different and will
use additional rounds for the recruitment process.}

For each leader, at the end of the recruitment phase, $\sqrt{\N}$ agents will have obtained a color
based on the original coin toss of that leader. We will refer to these $\sqrt{\N}$ agents as 
the \emph{cluster} associated with the leader, and we will sometimes describe agents as belonging to the
same cluster or different clusters. 
At the end of the recruitment phase and the beginning of the evaluation phase,
a constant fraction of the population will have been colored
having been recruited into the clusters of the various leaders, 
roughly half with color $0$ and the other half with color $1$.

Consider a particular agent in the evaluation phase.
If the agent meets another agent from the same cluster, then they necessarily
have the same color. If the agent meets an agent from a different cluster, then their colors are independently random.
Consequently, if the current population size is $m$, then with probability 
$\frac{1}{2}+\Theta(\frac{\sqrt{\N}}{m})$ the two agents will have the same color, and with probability
$\frac{1}{2}-\Theta(\frac{\sqrt{\N}}{m})$ the two agents will have different colors.
We can choose the splitting probability $p_\splt =1- \Theta(1/\sqrt{\N})$ so that the expected change in population
is zero for $m=\N$. 
For $m\gg \N$ the expected change in population will be negative, and for $m\ll \N$ the expected
change in population will be positive. 
Moreover, unlike in Attempt~2 above, which behaved similarly to a random walk,
here the effect is strong enough to maintain the population in a small interval around the target value $\N$,
with all but negligible probability.

Moreover, the adversary can do little to influence the result of the protocol. For a population size 
$m=\Theta(\N)$, the number of leaders selected is $\sqrt{\N}$, and so the standard deviation of 
the number of leaders with each color is roughly $\N^{1/4}$. Consequently, an adversary that
can insert or delete $o(\N^{1/4})$ agents can do little to influence the distribution of colors.
Even if the adversary selectively inserts or deletes agents that are leaders and have a particular color,
the effect of the adversary on the distribution of colors is dominated by the random deviation of the
sampling process. Unlike in Attempt~1 above, there are many leaders, and so the adversary is unable
to delete enough of them or to insert enough additional leaders to overwhelm the protocol. 
As we will show, the adversary cannot cause substantial deviations in the size of the population.

Note that one strategy the adversary may attempt is inserting agents that do not know the correct round number
within the epoch. 
In the protocol described so far, there is no mechanism for detecting and correcting this, and so over many rounds,
adversarial insertions may lead to a population of agents attempting to execute different portions of the protocol.
In order to address this, agents can exchange which round they are in, and self-destruct upon encountering
an agent that is in a different round of the epoch. This results in the self-destruction of any agent 
with the wrong round number as soon as it encounters an agent with the correct round number. A corresponding number
of correct agents are also destroyed, but we will show that the number of correct agents removed in this manner
is sufficiently small.   

As discussed briefly in Section~\ref{ssec:contributions}, we can think of this protocol as encoding
the population size in the variance of a distribution and then sampling from this distribution to obtain
a weak estimate of the variance. 
Since the variance of the fraction of successes in many independent Bernoulli trials decreases 
as the number of trials increases, 
if the number of leaders is larger, then the fraction of colored agents with color $0$ will
be more closely concentrated around $1/2$.
On the other hand, if the number of leaders is smaller, then we expect the fraction of colored agents with color $0$
to be farther from $1/2$.
Since the expected number of leaders is proportional to the current size of the population, 
an approximation to the population size is encoded in the fraction of agents of each color. 
Consider the distribution obtained by selecting an agent at random and reading its color. 
Comparing the colors of two agents serves as a very weak estimate of the variance of this distribution,
while aggregating the results of many agents' individual choices of whether to replicate or self-destruct
serves to amplify the accuracy of this estimate.

\paragraph{Achieving constant-size messages.} The protocol described above involves messages
of size $\Theta(\log\log \N)$ bits, essentially as large as the entire memory of an agent. We now outline how to modify
the protocol to use constant-size messages.  The only large portions of the messages described so far
consist of the current round in the epoch and the depth in the recruitment tree, each of which
can be encoded in $\Theta(\log\log\N)$ bits. 

The current round in the epoch is sent to prevent the adversary from confusing the protocol by inserting
agents with the wrong round number. However, rather than sending the exact round, we will instead
send the single bit specifying whether or not the agent is currently in the evaluation round. If an agent is entering
the evaluation round and its neighbor is not, then both agents will self-destruct. We can show
that this suffices to maintain the invariant that a large majority of the agents are in the same round of the epoch.

The depth of the recruitment tree is sent to allow each leader to induce the recruitment of the correct number
of agents. Note that we cannot simply recruit for $\frac{1}{2}\log \N$ rounds, since recruiting agents may encounter
other agents that are already colored and cannot recruit them. 
However, we can slow down the recruitment process to allow the depth in the recruitment tree
to be determined as a function of the round number. To do this, we recruit for $\Theta(\log^3 \N)$ rounds,
divided into $\frac{1}{2}\log \N$ subphases of $\log^2 \N$ rounds 
each.\footnote{We actually only require that subphases are $\omega(\log \N)$ rounds, so it is sufficient to recruit for
$\omega(\log^2 \N)$ rounds.}
In a single subphase, a recruiting agent will recruit only the {\em first} inactive agent it sees 
even if it encounters many inactive agents in the subphase. 
This allows agents to determine their depth in the recruitment tree based on the round in which they
were recruited.
Since a subphase consists of $\omega(\log\N)$ rounds, we will show that an inactive agent will be encountered
with high probability.

This yields a population control protocol with constant-size messages, since messages
consist of four binary values, namely an agent's color, whether or not it is active, whether or not it is recruiting, 
and whether or not it is currently in the evaluation round. In the analysis we will see how to achieve the
same result with only three-bit messages.  

\subsection{Related work}\label{related} %

\vspace{3pt}
\noindent
{\bf Population Protocols.}
The population protocol model was introduced by Angluin et al.~\cite{AADFP04,AADFP06}. In this model a collection of agents, which are modeled by finite state machines, move around unpredictably and have pairwise interactions.
The original definition considers a worst-case environment/scheduler, while later formulations~\cite{AAE07} consider the case where each interaction occurs between a pair of agents chosen uniformly at random.
In a population protocol, agents start with an initial configuration, and the goal is to jointly compute a function of this input.
Previous works have tried to identity the class of functions that can be computed in such a model~\cite{AAER07}, and the tradeoffs between the resources need to do so (e.g.~\cite{AAEGR17}). In these works, 
the agents are always active throughout the execution of the protocol.

Another line of work expands the population model to the case in which agents 
can crash or undergo transient failures that corrupt their states.
Delporte-Gallet et al.~\cite{DFGR06} consider a setting in which agents must compute a function of their inputs
in presence of such failures.
They construct a compiler that takes as input a protocol that works in the failure-free model, and outputs a protocol that 
works in the presense of failures as long as modifying a small number of inputs does not change the function output.
Angluin et al.~\cite{AAFJ08} incorporated the notion of self-stabilization into the population protocol
model, giving self-stabilizing protocols for some classical problems such as leader
election and token passing. They focus on the goal of stably maintaining some property such as having a
unique leader or a legal coloring of the communication graph.

Unlike these works, in our work agents have the ability to reproduce and self-destruct,
and the goal of maintaining a consistent population size must be carried out in the presence of an 
adversary with the corresponding capability to insert and delete agents.

\vspace{3pt}
\noindent
{\bf Approximate Counting.}
The problem of maintaining the population size of a collection of memory-constrained agents is related to the problem of counting
$N$ items when the available memory of the agents is less than $\log N$ bits.
Approximate counters were introduced by Morris~\cite{Mor78} as technique to accurately approximate a value $N$ using only
$\Theta(\log \log N)$ bits of memory. Techniques for approximate counting in the population model were developed in~\cite{ABBS16,AAEGR17}.

A sequence of works by Di Luna et al.~\cite{LBBC14,LBBC14b} consider the problem of estimating the size of a network of agents 
that communicate according to a dynamic connection graph, in presence of an adversary that can add and remove {\em edges} in the graph.

\vspace{3pt}
\noindent
{\bf Cellular Automata.}
Cellular automata were proposed by von Neumann~\cite{Neu} as a model to reason about artificial self-reproduction.
A cellular automaton consists of a regular grid of cells, each assuming one of a finite number of states.
Over time, the states of cells change according to some fixed rule (e.g.~a mathematical function) that determines the new state of each cell in terms of the current state of the cell and the states of the cells in its neighborhood.
Conway's Game of Life~\cite{GameofLife} is a cellular automaton that works with a simpler set of rules than von~Neumann's rules, and was shown to be Turing-complete by Berlekamp, Conway and Guy~\cite{BCG}.
Cook~\cite{Cook} proved that rule 110 (a binary, {\em one}-dimensional cellular automata) is Turing-complete.

Our setting is crucially different from the cellular automaton setting, since agents in our model
do not simply change state but can be deleted from the system during the computation.
Another difference between our setting and the cellular automaton setting is that we consider an {\em adversarial} model whereas in cellular automata cells deterministically change state.
In some sense, in Conway's game, death ``plays by the rules'' while in our game death is sudden and unpredictable.

\vspace{3pt}
\noindent
{\bf Dynamic Environments}.
A recent work of Goldreich and Ron~\cite{GR17} considers environments that evolve according to a fixed local rule. They define an {\it environment} as a collection of small components of a large system which interact in a local level, and change state according to a fixed rule. As an example, they focus on the model of a two-dimensional cellular automata. They ask how many queries a global observer must make about local components in order to test whether the evolution of the environment obeys a fixed known rule or to predict the state of the system at a given time and location.
Although their work seems very different than ours, it bears some intellectual similarity in seeking information about a global property of the system from local information.
However, whereas in \cite{GR17} a global observer who can query a limited number of individual cells asks ``does the global system obey a specific evolution rule,'' in our case
individual agents need to decide locally what they should do to maintain the overall global property of population size.

\vspace{3pt}
\noindent
{\bf Self-Stabilization.}
Also related to our question is the self-stabilization problem introduced by Dijkstra~\cite{Dj74}.
Given a system that starts in an arbitrary state, the goal of a stabilization algorithm is to eventually converge to the correct state. In this setting, however, deletion of system components is not considered.
Super-stabilization~\cite{DH97} is the problem of achieving self stabilization in dynamic networks, that is, network where nodes are dynamically added and removed. While this setting is closer to ours,
super-stabilization algorithms make additional assumptions about the system, such as that each
node in the system is uniquely identified.

\vspace{3pt}
\noindent
{\bf Distributed Algorithms Explaining Ant Colony Behaviors.}
A single ant has very limited communication and processing power, yet collectively a colony of ants can perform complex tasks such as consensus decision-making, leader election, and navigation.
In \cite{CDLN14} Cornejo et al.~give a mathematical model for the problem of task allocation in ant colonies,
and propose a very efficient protocol for satisfying this task. One of the main goals of their paper is to provide
a formal model enabling the comparison of the
various task allocation algorithms proposed in the biology literature.
Similarly, in~\cite{GMRL15} Ghaffari et al.~use techniques from distributed computing theory in order
to gain insight into the {\em ant colony house-hunting problem}, where a set of agents need to identify potential nests, evaluate the quality of candidates, and reach consensus in a distributed manner.

\section{The Population Stability Problem}\label{model}
As discussed above, the population stability problem is concerned with
a system of agents with bounded memory and the ability to reproduce and self-destruct.
The system is subjected to adversarial attacks that delete or insert processors. 
The objective is to maintain a stable population size despite these adversarial attacks. In this section we 
give a formal description of the problem.

\paragraph{Parameters.}
The population stability problem is parameterized by the initial number of agents $\nzero$, the number of distinct memory states 
$\agentmemory$ each agent can be in, the number of  alterations $\maxkill$ the adversary is allow to make in each round (where an alteration consists of removing or inserting an processor), a value $\alpha$ specifying how tightly concentrated the population size must remain around the target value $\nzero$, and a value $\gamma$
specifying a lower bound on the fraction of processors that are matched with other processors in each round.
Below we describe each of these components of the problem.

We note that one could consider separate parameters for the number of adversarial deletions and insertions, 
allowing the adversary to make a different number of each. 
In this paper we will consider both to be bounded by a single parameter $\K$.

\paragraph{Agents.}
We consider agents with bounded memory. %
Each agent can be in one of $\agentmemory$ possible states, so we can think of agents as having $\log \agentmemory$
bits of memory. Agents can communicate by message-passing 
as specified by the connectivity structure of the system, which will be discussed further below. 
In our setting we will have $\agentmemory \ll \nzero$, 
so each individual agent has insufficient memory to count the total population, to posses a unique ID,
or to address messages to a particular recipient.
Each agent has the ability to flip unbiased coins, %
to split into two identical agents,  or to self destruct. That is, in any round  an agent may choose to
split into two daughter agents which both inherit specified state from the parent agent, or it can decide to 
delete itself from the system.

\paragraph{Connectivity.}
As discussed above, we consider a {\em synchronous}  version of the population model of 
Angluin et al.~\cite{AADFP04,AAE07},
where we assume the existence of a global clock.
We assume that the pairs of agents that are able to communicate in each round
are selected by choosing a random matching of at least a $\gamma$ fraction of surviving agents. 
We think of the parameter $\gamma$ as a constant (e.g. $\gamma=1/4$).  
That is, each agent is matched with at most one other agent (that we call its \emph{neighbor}) in each round,
and there is no consistency from round to round, since connectivity in different rounds is determined by sampling 
independently random matchings. The schedule of these matchings is unknown to the adversary in advance.

\paragraph{Adversary.}
We consider a worst-case, computationally unbounded adversary that can arbitrarily choose which agents to delete in
each round and can insert agents with arbitrary state in each round. 
The adversary also can observe the entire history of agent interactions, including
the memory contents of every agent.
While the initial state of inserted agents is determined by the adversary, 
the newly inserted agents are assumed to follow the protocol 
(that is, the agents introduced by the adversary do not behave maliciously).

\paragraph{Objective.}
The goal in the population stability problem is to maintain a number of agents within a small interval 
$[(1-\alpha)\nzero, (1+\alpha)\nzero]$ around the initial population size $\nzero$. That is, initially the system consists of 
$\nzero$ agents. In each round some agents may be removed or inserted by the adversary, and some agents may decide to 
replicate or to self-destruct. 
Let $N_i$ denote the number of agents in the system after the $i$th round. The adversary wins in round $i$ if
$N_i \notin [(1-\alpha)\nzero, (1+\alpha)\nzero]$ at the end of the round. 
We say that protocol $\Pi$ is a population stability protocol %
if for any polynomial $p$ and any adversary, the probability that the adversary wins
in at most $p(\N)$ rounds is negligible in $\N$.

\section{The protocol}
\label{sec:protocol}

We now provide a formal specification of the main protocol (Algorithm~\ref{alg:growth-main}).
Agents continually run the protocol throughout their lifetime. We will think of time as partitioned into epochs
of $\T = \frac{1}{2}\log^3 \N$ rounds.
Each epoch consists of three phases,
the leader selection phase, the recruitment phase, and the evaluation phase.
The recruitment phase consists of $\frac{1}{2}\log \N$ subphases each consisting of roughly
$\Tinner = \log^2 \N$ rounds.\footnote{More generally, we want 
$\T=\Tinner \cdot \frac{1}{2}\log \N$ for any $\Tinner = \omega(\log \N)$. The first and last subphase
will each be shorter by one round to account for the leader selection and evaluation phases.} 
We will elaborate on each of these phases below. 

Recall that agents have the ability to toss coins, to send and receive a message 
upon encountering another agent, %
to reproduce by splitting into two identical copies of itself, and to self-destruct. These capabilities are
notated by the following functions. We denote flipping an unbiased coin by $x\rgets\binset$. 
Agent splitting and death are implemented in commands 
$\Call{Split}$ and $\Call{Die}$.
Finally, the command $Z := \Call{\Communicate}{X}$ sends message $X$ to 
the neighboring agent in the present round, if any,
simultaneously receiving in response message $Z$.\footnote{Recall that an agent's
\emph{neighbor} is the other agent the agent
is randomly matched with in this round and that matchings in each round are independent 
and uniformly random.}
If the agent is unmatched in this round and
has no neighboring agent, then the return value is assumed to be $\bot$. 
In the protocol below, messages will consist of four boolean values 
$(\SelfIsER,\amActive,\mycolor,\recruiting)$. Upon receiving message $Z$, the value of
each of these variables can be accessed by writing $Z.\SelfIsER$, $Z.\amActive$, and so forth. If $Z=\bot$ then 
we will follow the convention that each of these components will also have value $\bot$. 

The main variables that describe the state of an agent are $\round \in [0,\T-1]$ and four boolean
variables $\amActive\in\binset$
$\mycolor \in \binset$, $\recruiting\in\binset$, and $\SelfIsER\in\binset$. The variable 
$\NbrStatus\in\binset^4$
stores the most recent message received, consisting of four boolean values
$(\Nbr.\SelfIsER,\Nbr.\amActive,\Nbr.\mycolor,\Nbr.\recruiting)$.
An additional variable, $\torecruit \in [0,\frac{1}{2}\log{\N}]$, is not necessary for the protocol itself but
is used in the analysis. We emphasize
that these variables are local to a single agent, so that different agents have independent
copies of each of these variables. 
Initially, at the onset of the system, for each agent we will have 
that all variables are set to zero.

The variable $\round$ keeps track of which round it is within the epoch. The variable
is incremented modulo $\T$ after each round, 
ensuring that agents begin and end each phase and each epoch at the same time. 

The variables $\amActive$ and $\mycolor$ specify whether an agent has been activated and colored,
as well as the color of the agent.
In the first round of each epoch,
some of the agents will designate themselves as leaders and will become active,
choosing at random 
a color $\mycolor \rgets \binset$, while the rest of the agents remain inactive.
During the recruitment phase, additional agents will become active and will receive colors in 
$\binset$, as inherited from a leader. The value of variable $\mycolor$ is only relevant
for active agents.

The variable $\recruiting$ specifies whether or not an active agent is trying to recruit in the present
subphase. Each active agent should recruit only one additional agent in a single subphase, so
this variable specifies whether or not the agent is still looking for an inactive agent to recruit in
the subphase.

The variable $\SelfIsER$ specifies whether the agent is currently in the evaluation phase, which is true exactly when 
$\round = \T-1$. 

Finally, the variable $\torecruit$ specifies the number of additional followers
an active agent is tasked with recruiting directly, which is the logarithm of the total number
of agents that should be activated as a result of the given agent. 
When a new leader first becomes active, it sets $\torecruit = \frac{1}{2}\log\N$, indicating
that it is tasked with recruiting a total of $2^{\torecruit} = \sqrt{\N}$ agents.
Each time a new agent is recruited, the value of $\torecruit$ is decremented and shared
with the newly recruited agent, indicated that each of the two is responsible for recruiting only
half of the total. For instance, after the first time a leader recruits another agent,
both agents will have $\torecruit = \left(\frac{1}{2}\log \N\right) - 1$, and so each
of the two agents subsequently are responsible for recruiting only $\sqrt{N}/2$ agents. 
In this way a leader can induce the recruitment of $\sqrt{\N}$ agents in roughly a logarithmic number of rounds.
Although the variable is not used by the algorithm itself, we will refer to it in the analysis.

We first present the main procedure run by each agent in every round. Each agent
first exchanges messages with its neighboring agent in this round, if any. 
The agent then performs a consistency check on its on state and the state of its neighbor.
Then, depending on the value of variable $\round$ modulo $\T$, the program
calls the appropriate subroutine for the corresponding phase. 
\begin{algorithm}[H]
\caption{Main protocol}
\label{alg:growth-main}
\begin{algorithmic}[1]
\Procedure{MainProtocolStep}{}
\State $\Call{\ExchangeMessages}$
\State $\Call{\CheckRoundConsistency}$
\Statex

\If {$\round=0$} \Comment{Handle initialization of leaders}
  \State $\Call{\ElectLeader}$
  \State $\round := \round + 1$
\ElsIf {$\round < \interval - 1$} \Comment{Perform recruitment phase}
  \State $\Call{\RecruitmentPhase}$
  \State $\round := \round + 1$
\Else \Comment{Final round of phase. Perform split or death.} 
  \State $\Call{\EvaluationPhase}$
  \State $\round := 0$
\EndIf
\EndProcedure
\end{algorithmic}
\end{algorithm}

We now describe the subroutine that exchanges messages with the neighboring agent, if any.
An agent simply computes the indicator value of whether it is in the evaluation phase,
and sends this information along with its activation state, color, and recruiting status. 
It receives a corresponding message from its neighbor, or the value $\bot$ if 
it has no neighbor in this round.

\begin{algorithm}[H]
\caption{Subroutine to send and receive messages with neighboring agent}
\label{alg:gro:send-and-receive}
\begin{algorithmic}[1]
\Procedure{\ExchangeMessages}{} %
\If { $\round = \interval-1$ }
  \State $\SelfIsER := 1$
\Else
  \State $\SelfIsER := 0$
\EndIf
\State $\MyStatus := (\SelfIsER,\amActive,\mycolor,\recruiting)$
\State $\NbrStatus := \Call{\Communicate}{\MyStatus}$
\EndProcedure
\end{algorithmic}
\end{algorithm}

We now describe the leader selection subroutine, which comprises the first round of each epoch.
This is an entirely non-interactive process in which each agent becomes a leader with some fixed probability
$1/(8\sqrt{N})$ by tossing its own coins, entirely independently of each other agent.
With overwhelming probability, if the total number of agents is $\Theta(\N)$, the number of leaders
chosen in this phase will be $\Theta(\sqrt{\N})$. Each newly activated
leader chooses a random color in $\{0,1\}$ and is tasked with recruiting $\sqrt{\N}$ agents and
assigning them this color.

\begin{algorithm}[H]
\caption{Leader selection phase subroutine}
\label{alg:gro:leader-election}
\begin{algorithmic}[1]
\Procedure{\ElectLeader}{} %
  \State $\amActive := \Call{\BiasedCoin}{\log(8\sqrt{\N})}$.
  \Statex\Comment{$\amActive=1$ with probability $1/(8\sqrt{\N})$, and otherwise $\amActive=0$.}
  \If {$\amActive = 1$} 
    \State $\mycolor \rgets \{0,1\}$ 
    \State $\recruiting := 1$
    \State $\torecruit := \frac{1}{2} \log(\N)$
  \EndIf
\EndProcedure
\end{algorithmic}
\end{algorithm}

The leader selection phase, as well as the evaluation phase below,
requires the ability to flip biased coins with bias $1/\poly(\N)$, in particular
$1/\Theta(\sqrt{\N})$, where bias refers to the probability of the coin having value $1$
(so that a fair coin has bias $1/2$). Recall that we assume only the ability to toss unbiased coins.
We now give a simple procedure to obtain the desired bias using
only $O(\log\log(\N))$ bits of memory, assuming that $\log \N$ is an even integer. More generally, we
show how to obtain bias $2^{-a}$ for any integer $a$ using $1+\lceil\log a\rceil$ bits of memory. 
We note that it is sufficient to toss $a$ coins and report $1$ if they all landed heads
and $1$ otherwise. This requires counting to $a$, which can be done with $\log a$ memory.

\begin{algorithm}[H]
\caption{Subroutine to toss a biased coin that equals $1$ with probability 
$2^{-a}$ and equals $0$ otherwise}
\label{alg:gro:biased-coin}
\begin{algorithmic}[1]
\Procedure{\BiasedCoin}{a} \Comment{Flip a biased coin with bias $\Pr[c=1]=2^{-a}$}
\State $c := 1$
\For {$i=1$ to $a$ }
 \State $b \rgets \{0,1\}$ 
 \If {$b=0$}
    \State $c := 0$
 \EndIf
\EndFor
\State \Return $c$
\EndProcedure
\end{algorithmic}
\end{algorithm}

We now describe the recruitment phase, which is the second phase executed by each agent in every epoch
and is the main source of interaction in the protocol. The phase lasts for $\T=\Theta(\log^3 \N)$ rounds,
consisting of $\frac{1}{2}\log \N$ subphases each of length $\Tinner = \Theta(\log^2\N)$ rounds.
As discussed above, during this phase
 each leader is tasked with finding $\sqrt{\N}$ inactive agents (i.e. with $\amActive = 0$)
and coloring each of them with the color of the leader. We note again that the leader will not directly meet
each of these $\sqrt{\N}$ inactive agents, but rather that this is done by propagation, where the leader
will activate some agents and each of these will activate additional agents.
In each subphase each active agent will attempt to recruit a single nonactive agent, which will
then start to recruit in the following subphase. Since there are $\frac{1}{2}\log \N = \log\sqrt{\N}$ subphases,
if each attempt to recruit is successful, then a single leader will result in the activation of a total of $\sqrt{\N}$ agents.
\begin{algorithm}[H]
\caption{Recruitment phase subroutine}%
\label{alg:gro:recruitment-phase}
\begin{algorithmic}[1]
\Procedure{\RecruitmentPhase}{} %
\If {$\NbrActive = 0$ and $\recruiting = 1$} \Comment{Other agent has been activated}
    \State $\recruiting := 0$
    \State $\torecruit := \torecruit - 1$
\ElsIf {$\amActive = 0$ and $\NbrRec = 1$}  \Comment{This agent must be activated}
    \State $\amActive := 1$
    \State $\mycolor := \NbrColor$
    \State $\recruiting := 0$
    \State $\torecruit := \frac{1}{2}\log \N - \left\lceil (\round+1) / \Tinner \right\rceil$
  \EndIf
  
\If {$\round \equiv -1 \pmod{\Tinner}$} \Comment{Prepare for next round}
  \State $\recruiting := 1$.
\EndIf

\EndProcedure
\end{algorithmic}
\end{algorithm}

The final phase of the algorithm is the evaluation phase, which occurs on the last round of each epoch.
In this phase, each matched active agent compares its color to that of its neighbor and makes a decision
of whether to replicate itself or to self-destruct. 

\begin{algorithm}[H]
\caption{Subroutine to perform splitting and self-destruction at the end of each phase}
\label{alg:gro:evaluation-phase}
\begin{algorithmic}[1]
\Procedure{\EvaluationPhase}{} 
  \If {$\amActive=1$ and $\NbrActive=1$} 
    \If {$\NbrColor = \mycolor$} \Comment{Colors same: split with probability $1-16/\sqrt{\N}$}%
      \State $c := \Call{\BiasedCoin}{\log(\sqrt{\N}/16)}$
      \If {$c=0$} 
        \State $\Call{Split}$ %
      \EndIf
    \Else \Comment{Colors different: self-destruct with probability $1$} %
      \State $\Call{Die}$ 
    \EndIf
  \EndIf 
  \State $\amActive := 0$
  \State $\mycolor := 0$
  \State $\recruiting := 0$
\EndProcedure
\end{algorithmic}
\end{algorithm}

Finally, we give the subroutine invoked at the very beginning of each round,
that performs a consistency check on the
round values of the agent and its neighbor. 
In the absence of adversarial insertions this subroutine is unnecessary, since
agents will always have the correct round value in the epoch. However, it is necessary if the adversary
is allowed to insert agents with an incorrect round value. If left to increase unchecked, the presence
of many agents with different round values would interfere with the operation of the protocol.
We will prevent this from happening by causing agents to self-destruct
as soon as they encounter an agent with a different round value. 
However, implementing this exactly would require $\Theta(\log\log\N)$-bit  messages, 
since agents would need to exchange their round numbers. To avoid this, we will instead
have agents exchange only an indicator variable for whether or not they are in the evaluation
phase. An agent will self-destruct if it is in the evaluation phase and its neighbor is not,
or if its neighbor is in the evaluation phase and it is not. 
This process deletes a small number of agents with the correct round number along with agents
with the incorrect round number. We will show in the analysis below that this will ensure that there
are few agents with the incorrect round number and that only a small number of agents
with the correct round number will self-destruct as a result of this procedure.

\begin{algorithm}[H]
\caption{Subroutine to determine if agent knows the correct round}
\label{alg:gro:round-consistency}
\begin{algorithmic}[1]
\Procedure{\CheckRoundConsistency}{} %
\If { $\Nbr \neq \bot$ and $\SelfIsER \neq \NbrIsER$ }
  \State $\Call{Die}$
\EndIf
\EndProcedure
\end{algorithmic}
\end{algorithm}

\section{Analysis}
\label{sec:analysis}
In this section we prove the main theorem. 

\begin{theorem}
Let $\alpha,\gamma,\epsilon$ be positive constants, where $\gamma\leq 1$ is a lower
bound on the fraction of agents that is matched in each round.
Then Algorithm~\ref{alg:growth-main} is a population stability 
protocol using $\omega(\log^2 \nzero)$ states
per agent and three-bit messages\footnote{A straightforward implementation of the protocol described above
would use $\Theta(\log^4 \nzero)$ states and four-bit messages. We describe below how to achieve the improved bounds 
stated here.}
guaranteeing that if the adversary inserts and deletes at most $\maxkill = O(\N^{1/4-\epsilon})$
agents in each round, then with
all but negligible probability the population will remain between $(1-\alpha)\nzero$ and $(1+\alpha)\nzero$ 
for any polynomial number of rounds. 
\label{thm:mainthm-restated}
\end{theorem}

For ease of presentation, the version of the protocol described above uses four-bit messages. 
However, we can reduce the message size to three bits, as follows. 
If the agent is in the evaluation phase (i.e.~$\SelfIsER=1$) then the message must contain the values 
$\amActive$ and $\mycolor$, but need not contain $\recruiting$. 
If $\SelfIsER=0$, then the message must contain the value $\mycolor$ but not $\amActive$ if $\recruiting=1$
and the value $\amActive$ but not $\mycolor$ if $\recruiting=0$. 
Consequently, the desired information can be encoded in only three bits. 

For the memory requirements, $\log T$ bits are needed to store the variable $\round\in \binset$, and the other
variables stored by each agent consist of eight boolean values. The invocations of the subroutine
$\Call{\BiasedCoin}$ require $\log \log N$ bits of local memory. 
However, the subroutine is only invoked in two rounds of each epoch, the  leader selection round and the evaluation round.
Consequently, using additional indicator bits to specify whether an agent is in each of these those two rounds, the memory used to 
store the variable $\round$ can be used as the helper memory for the subroutine, and so additional memory is
not necessary. For $T=\frac{1}{2} \log^3 \N$, the total
number of states is therefore $\Theta(\log^3 \N)$. However, it suffices to have $T=\Tinner\cdot\log \N$ for any
$\Tinner = \omega(\log \N)$, and so we can reduce the number of states to $\omega(\log^2 \N)$. 

\paragraph{Roadmap to proof}
We must show that the population size will remain close to the target value $\N$.
We will do this by means of two key steps.

The first step is to show that in any single epoch the population size will be relatively stable.
That is, we show that the population in the middle or end of an epoch
will not be too much larger or smaller than the population at the beginning of the epoch. 
We achieve this by showing that irrespective of the adversary's actions, with overwhelming probability
the number of agents of each color in the evaluation phase will be concentrated around one-sixteenth of
the total number of agents.
This step is formalized in Lemmas~\ref{lem:agentcounts} and \ref{lem:bounded_deviation} in 
Section~\ref{ssec:bounded-deviation}.

The second step is to show that the population size will tend to correct itself if it has deviated too far
from the target value $\N$.
More precisely, we show that if the population is far from $\N$, 
then in expectation the population at the end of the phase will be substantially
closer to $\N$ than the population at the start of the phase.
This step highlights a key tension of the proof, namely the difficulty analyzing a system
that contains both random components in the matching schedule and worst-case components
in the adversary's insertions and deletions. Indeed, it is not even clear a priori what it means
to discuss expectation in a system with a worst-case adversary.
This step is formalized in Lemma~\ref{lem:correct_in_exp}
in Section~\ref{ssec:correcting-drift}.

Putting these two steps together will enable us to conclude the proof of the theorem. 
Consider the first epoch in which the population size lies outside the interval
$[\omha\N,\opha\N]$. Since the population size does not deviate too much in a single epoch, 
the population at the start of the next epoch will be close to $\opmha\N$. 
But for each epoch in which the population remains outside the interval $[\omha\N,\opha\N]$,
in expectation the change in population will be in the direction of the target value $\N$. 
Considering the next $\N^{0.01}$ epochs, a Chernoff-Hoeffding bound then implies that with overwhelming probability
the population will return to the interval $[\omha\N,\opha\N]$. Since the population
does not change much in each epoch, it follows further that the population will remain inside the interval
$[\oma\N,\opa\N]$ during these $\N^{0.01}$ epochs.  
We will conclude that with all but negligible probability, the population will remain in the interval $[\oma\N,\opa\N]$
for any polynomial number of rounds. 

We now outline the remainder of the section.
In Section~\ref{ssec:bookkeeping-lemmas} we prove some preliminary lemmas about the protocol.
These lemmas provide us with invariants that we will need in order to prove the two key steps above.
In particular, we show that nearly every agent knows the correct round in the epoch,
that at least half of the agents are inactive (i.e. have $\amActive=0$) at any point
in the execution of the protocol, and that any leader selected in the first round of the epoch
will succeed in recruiting a full cluster of size $\sqrt{\N}$ unless either the adversary deletes some agent
in that cluster or an agent with the wrong round number interferes with the recruitment.

In Section~\ref{ssec:bounded-deviation} we prove the first of the key steps, showing that
with high probability the population size does not deviate too much in a single epoch.
In Section~\ref{ssec:correcting-drift} we prove the second of the key steps, showing that if the population
has drifted too far from the target value, then in expectation the population will correct itself. 
Finally, in Section~\ref{ssec:finish-the-proof} we conclude the proof
of the main theorem.

\subsection{Bookkeeping lemmas} 
\label{ssec:bookkeeping-lemmas}

We will first prove several bookkeeping lemmas to guarantee
that with overwhelming probability, certain invariants continue to hold throughout
the execution of the protocol. We will subsequently use these invariants to prove 
stronger statements that will enable us to conclude the correctness the protocol.

Recall that during the protocol, agents keep track of the current round within the epoch, that is,
the round number modulo $\T=\log^3\N$. However, the adversary is empowered to insert 
new agents into the system with arbitrary initial state, and in particular may insert agents with the
incorrect round number. 
Our first lemma provides a bound on the number of agents with the incorrect round number.
Recall that $\gamma=\Theta(1)$ is a lower bound on the fraction of parties in each round that are matched. 
We can assume that $\alpha\leq 1/2$, since the desired statement is stronger for
smaller $\alpha$.

\newcommand{\Nmin}{\N_{\mathsf{min}}}
\begin{lemma}
\label{lem:correctround}
For any $t>0$, suppose that the population size remains above $\Nmin = \N/2$
for the first $t$ rounds of the protocol. Then there exists some negligible function $\negl$ such that
conditioning on this event,
with probability $1-t\cdot \negl(\N)$,
all but $(1+\gamma^{-1})\N^{1/4}$ of the agents will 
have the same value for variable $\round$ in each of these $t$ rounds.
\end{lemma}

\begin{proof}
We prove this by induction on the round number. Initially all agents have $\round=0$.
Assume for induction that at round $r$ all but $\gamma^{-1}\N^{1/4}$ of the agents
have variable $\round=0$. We will show that with high probability, the same statement holds
at the start of round $r+\T$. The adversary may add an additional $\K$ agents in each of these
$\T$ rounds, for a total of $\K\T \leq \N^{1/4}/8$ agents.
Note that each agent with the wrong round value may split at most once in this epoch of $\T$ rounds,
when it reaches its evaluation phase. 
If there are $v=O(\N^{1/4})$ agents which differ from the majority value for variable $\round$,
then the probability of such an agent being matched with another such agent in its evaluation phase  
is at most $\gamma \cdot \frac{v}{\Nmin} = \gamma\cdot O(\N^{-3/4})$.
It follows that with high probability, the number of agents with the wrong round value that split 
during this epoch is at most $\N^{0.01}/8<\N^{1/4}/8$. 
Consequently at any point in this epoch of $\T$ rounds, the number of agents with the wrong 
round value will be at most $\gamma^{-1}\N^{1/4} + \N^{1/4}/4$.
Each agent with the wrong round value at the start of the majority evaluation phase 
(i.e. round $r+\T-1$) has probability at least 
$\gamma (\Nmin - (\gamma^{-1}+1/4) \N^{1/4}) / (\Nmin) \geq \gamma - 3\N^{-3/4}$ of being matched with an 
agent starting the evaluation phase,
so with all-but-negligible probability, at most 
$(1/\gamma - 1/2)\N^{1/4}$ of these agents will not be matched with an agent with
different $\round$ value and will survive the round. It follows that at the start of 
round $r+\T$, at most $\gamma^{-1}N^{1/4}$ agents have $\round$ value different from $0$.
Consequently, by induction we have that with overwhelming probability, the number of agents with variable 
$\round \neq 0$ in any round $r\equiv 0 \pmod{\T}$ is at most $\gamma^{-1} \N^{1/4}$.
Since we showed above that the number of agents with the wrong round number
that can be added during the epoch is small with overwhelming probability,
the lemma follows.
\end{proof}

\begin{lemma}
\label{lem:halfactive}
With high probability, if the population is in the interval $[\oma N, \opa N]$ at the start of
an epoch, at any point in the epoch, at most $1/2$ of the agents have $\amActive = 1$. %
\end{lemma}
\begin{proof}
Let $m\in[\oma \N,\opa \N]$ be the population at the start of the epoch.
With all but negligible probability the number of leaders chosen will be 
$m / (8\sqrt{\N}) \pm o(\N^{0.01})$ %
by a Chernoff-Hoeffding bound.
Each leader may induce the activation of at most $\sqrt{\N}$ total agents. 
During the epoch, the adversary may insert an additional 
$\T\cdot \advbound  \leq \N^{1/4}$ %
agents, which may each 
induce the activation of $\sqrt{\N}$ total agents.
Consequently at any point in the epoch, the number of active agents
will be at most $m/8 + \N^{3/4}$. 
Prior to the evaluation step, the adversary can have killed at most 
$\T\cdot\advbound \leq \N^{1/4}$ agents. By the previous lemma, at most
$\gamma^{-1}\N^{1/4}$ agents at the start of the epoch can have the wrong value
for $\round$, and at most $\T\cdot\advbound$ additional such agents can be introduced
during the protocol, so at most $O(\gamma^{-1} \N^{1/4}) = O(\N^{1/4})$ agents 
can be killed in the $\Call{\CheckRoundConsistency}$ procedure. Consequently the population
throughout the epoch until the evaluation phase will be at least $m - O(\N^{1/4})$.
The conclusion follows.
\end{proof}

\begin{lemma}
\label{lem:recruitmentcompletes}
Suppose the population is in the interval $[\oma N, \opa N]$ at the start of
an epoch. Then with high probability, in the last round of the epoch, every active agent entering
the evaluation phase that was not inserted by the adversary during this epoch will have $\torecruit = 0$.
\end{lemma}
\begin{proof}
By Lemma \ref{lem:halfactive}, at most 
half of the agents are active in each round, so the probability in each round of encountering an inactive
agent at each round is at least $\gamma/2 = \Theta(1)$. 
With overwhelming probability, in any sequence of $\omega(\log \N)$ steps an agent will encounter
an inactive agent and will be able to recruit it. Applying a union bound, we have that in each cycle of $\Tinner$ steps, 
each of the $O(\N)$ active agents attempting to recruit will be successful in finding an inactive agent to recruit.
Consequently each agent will be able to recruit the desired number of additional agents, and the lemma follows.
\end{proof}

\subsection{Bounded deviation}
\label{ssec:bounded-deviation}

In this section we show that with high probability, the population size does not change by too much in
any single epoch.

\begin{lemma}
\label{lem:agentcounts}
Let $m$ be the population at the start of an epoch. With high probability,
the number of agents with each color $b\in\{0,1\}$ at the start of the evaluation phase
will be $m/16 \pm O(\N^{3/4-\epsilon})$. 
\end{lemma}
\begin{proof}
Let $m$ be the population at the start of the epoch.
With all but negligible probability, the number of leaders selected with color $0$
at the beginning of the epoch will be $m / (16\sqrt{\N}) \pm o(\N^{0.01})$,
and similarly for the number of leaders selected with color $1$. In the absence of
adversarial deletions, each leader will recruit $\sqrt{N}$ followers with the same coin value,
inducing the presence of $m / 16 \pm o(\N^{0.51})$ agents of each color by the final round
of the epoch. 

The adversary may insert or delete $\K\cdot \T = \tilde{O}(\N^{1/4-\epsilon})$ agents over the course of the epoch.
Each inserted agent can induce the activation of at most $\sqrt{\N}$ additional agents,
and similarly each removed agent could have activated up to $\sqrt{\N}$ additional agents.
Additionally, by Lemma \ref{lem:correctround}, at most $O(\N^{1/4})$ agents will
be removed in procedure $\Call{\CheckRoundConsistency}$ upon encountering
an agent with a different value for variable $\round$. 
Overall the actions of the adversary can affect the number of agents of each color by 
at most $O(\N^{3/4 - \epsilon})$. It follows that despite adversarial action,
the number of agents of each color at the start of the 
evaluation phase will be $m/16 \pm O(\N^{3/4-\epsilon})$ with all but negligible probability. 
\end{proof}

\begin{lemma}
\label{lem:bounded_deviation}
With all but negligible probability, if the population is in the interval $[\oma \N, \opa\N]$ 
at the start of an epoch, the population will have deviated by at most $\tilde{O}(\sqrt{N})$ 
by the end of the epoch.
\end{lemma}

\begin{proof}
Let $m\in[\oma \N,\opa \N]$ be the population at the start of the epoch.
By the previous lemma, at the start of the evaluation phase the number of agents with color $0$
will be $m_0 = m/16 \pm O(\N^{3/4-\epsilon})$ with all but negligible probability,  
and likewise the number of agents with color $1$
will be $m_1 = m/16 \pm O(\N^{3/4-\epsilon})$. We condition on these events. 
Since the adversary can insert at most $\K$ agents in each of the $\T=\log^3 \N$ rounds
for a total of $\tilde{O}(\N^{1/4-\epsilon})$, 
and no other new agents with the correct value of variable $\round$ can be produced
until the evaluation phase, Lemma \ref{lem:correctround} implies that
the total number of agents at the start of the evaluation phase is 
at most $m + O(\gamma^{-1} \N^{1/4})$.
Similarly, since the adversary can have directly removed at most $\K\cdot \T =\K \log^3 \N$ agents,
and at most $2((\gamma^{-1}+1)N^{1/4} + \K\log^3 \N)$ agents with the correct value of
$\round$ may have been removed as a result of procedure $\Call{\CheckRoundConsistency}$
after encountering an agent with a different value of $\round$,
it follows that for $\gamma=\Theta(1)$, the total number of agents at the start of the evaluation phase is
$m \pm O(\N^{1/4})$.

Consequently the communication graph for the evaluation phase is 
a random matching of size $q=\Omega(\gamma \N)=\Omega(\N)$.
Sample such a matching
by first choosing a set of $q$ left vertices
and a set of $q$ right vertices, and then associating corresponding vertices on the 
left and right. 
Let $q_0 = qm_0/m'$ and $q_1 = qm_0/m'$.
With high probability the number of left (respectively, right) vertices with color $0$ 
will be %
$q_0 \pm \tilde{O}(\sqrt{\N})$,
and similarly for vertices of color $1$. 

It follows that with all but negligible probability, the number of left-vertices of color
$b$ that split after being matched with a right-vertex of color $b$ is
$q_b^2/q \pm \tilde{O}(\sqrt{\N})$ for each $b\in \{0,1\}$.
Similarly, the number of left-vertices of color $b$ that self-destruct after 
being matched with a right-vertex of color $1-b$ is
$q_bq_{1-b}/q \pm \tilde{O}(\sqrt{\N})$.
Consequently will all but negligible probability,
noting that $q_0-q_1 = O(\N^{3/4})$, we have that
the change in population during the evaluation phase is
$$\frac{2}{q}\cdot \left(q_0^2 + q_1^2 - 2q_0q_1\right) \pm \tilde{O}(\sqrt{\N})
= \tilde{O}(\sqrt{\N})$$
as desired.

\end{proof}

\subsection{Correcting population drift} 
\label{ssec:correcting-drift}

In this section we show that if the population has drifted too far from the target value $N$,
in expectation it will tend to correct itself. 

\begin{lemma}
\label{lem:correct_in_exp}
If the population is in the interval $[\oma\N, \omha\N]$ at the start of an epoch,
then for any adversarial strategy, 
in expectation the population will increase by $\Omega(\sqrt{\N})$ by the
end of the epoch.
If the population is in the interval $[\opha N, \opa N]$ at the start of an epoch,
then
in expectation the population will decrease by at least $\Omega(\sqrt{\N})$ by the
end of the epoch.
\end{lemma}

\begin{proof}
The behavior of the system during the evaluation phase depends on the distribution of 
coin values of active agents in this phase. We would like to argue that each pair of clusters 
has its colors assigned by independent, fair coin flips. 
This is clearly true in the absence of an adversary. However, in our setting, adversarial 
instertions and deletions can bias the joint distribution of the colors of a pair of agents in
different clusters.\footnote{For instance, in the first round of the epoch, 
the adversary can instert additional leaders
all with color $0$, or can delete several leaders that have color $1$.
This difficulty arises because we allow the adversary to observe the internal memory of all
agents, including the results of coin tosses.}
Nonetheless, we will argue that the adversary's influence is limited, and that for most pairs of agents
in different clusters, we can think of the joint distribution of colors as unbiased and uniform.
We will label clusters as \emph{honest} or \emph{adversarial},
where the colors of a pair of honest clusters can be regarded as independently sampled,
and no assumption is made on the colors of the adversarial clusters.

Consider any fixed adversarial strategy. Note that the adversary can insert or delete a total of no more than 
$2\cdot \T\cdot\K = \tilde{O}(\N^{1/4-\epsilon})$ agents during the epoch, and consequently
can influence no more than this many clusters. 
Any strategy of the adversary during this epoch
can be emulated by deferring any deletions
of colored agents (with the correct $\round$ value) to the beginning of the evaluation phase, 
deleting instead an inactive agent. 
Recall that  by Lemma~\ref{lem:recruitmentcompletes}, each cluster not affected by the adversary
will consist of $\sqrt{\N}$ agents at the beginning of the evaluation phase.
At the beginning of the evaluation phase, we allow this
new adversary not only to delete the specified agent, but also to set $\amActive=0$ for any
subset of the $\sqrt{\N}-1$ other agents in the cluster. Note that any attack that could be accomplished
by the original adversary can still be carried out by this new adversary that defers all of its
deletions of colored agents to the evaluation phase but is subsequently allowed to modify up to 
$\tilde{O}(\N^{3/4-\epsilon})$ agents in $\tilde{O}(\N^{1/4-\epsilon})$ different clusters.

Since we now defer deletions of colored agents to the beginning of the evaluation phase,
each cluster induced by a leader selected in the first round of the epoch will have the full
$\sqrt{\N}$ members, and consequently these clusters are indistinguishable except for their color. 
Consequently as long as the adversary can modify agents in $\tilde{O}(\N^{1/4-\epsilon})$ clusters
of each color, \emph{which} specific cluster of each color is irrelevant.
Consider an arbitrary indexing of the agents before the first round of the epoch, and consider the
first $\tilde{O}(\N^{1/4-\frac{1}{2}\epsilon})$ leaders chosen in this round. 
With all but negligible probability, this set will contain at least $\tilde{O}(\N^{1/4-\epsilon})$
agents of each color $\mycolor\in\binset$, and so the strategy of the adversary can be carried
out by manipulating only the clusters of agents in this set. 
Let the clusters induced by these agents be the \emph{adversarial clusters} along with the clusters
induced by any agents inserted by the adversary, and let the remaining clusters be the 
\emph{honest clusters}. 

Now we have reduced to a setting in which we have achieved the 
desired property that each honest cluster has size $\sqrt{\N}$, and
the coin flips of any two honest clusters are independent. However, we are now dealing
with a modified adversary that can affect a larger overall number of agents.
Let $m'$ be the number of agents at the start of the evaluation phase.
By Lemma~\ref{lem:agentcounts}, it follows that with all but negligible probability,
the number of agents in honest clusters is $m_h = m'/8 \pm O(\N^{3/4-\frac{1}{2}\epsilon})$, and the
number of agents in clusters influenced by the adversary is $m_a = O(\N^{3/4-\frac{1}{2}\epsilon})$. 

Pick a random pair of vertices at the start of the evaluation phase. Then with probability
$1/64 - O(\N^{3/4}/m') = \Omega(1)$ both agents will belong to honest clusters, with probability 
$o(\N^{3/4}/m')$ one agent will belong to an honest cluster and the other to an adversarial cluster,
and with probability $o(\N^{3/2}/m'^2)$ both will belong to adversarial clusters.

A pair of vertices belonging to honest clusters will have the same color if they belong to the
same cluster, and independently random colors if they belong to different clusters. Consequently
the probability that such a pair of vertices will have the same color is
$\frac{1}{2} + \frac{\sqrt{\N}}{2m_h}-O(\frac{1}{m_h})$. 
It follows that the expected change in population resulting
from the matching of a pair of vertices belonging to honest clusters is
$$\left(1+\frac{\sqrt{\N}}{m_h}\right) \cdot \left(1-\frac{16}{\sqrt{\N}}\right) - \left(1-\frac{\sqrt{\N}}{m_h}\right)
-O\left(\frac{1}{m_h}\right)
= \frac{2\sqrt{\N}}{m_h} - \frac{16}{\sqrt{\N}} - \frac{O(1)}{m_h}\,.$$

Recalling that $\alpha\leq 1/2$ is a fixed constant,
for $m < \omha \N$ %
this quantity is $\Theta(1/\sqrt{\N})$, and for
$m > \opha \N$ this quantity is negative, with magnitude $\Theta(1/\sqrt{\N})$. 

The honest clusters consist of nearly the same number of agents of each color 
($m/16 \pm O(\N^{3/4})$ of each). It follows that the expected change in population resulting from
the matching of a vertex in an honest cluster and a vertex in an adversarial cluster has magnitude
$O(\N^{-1/4})$. Making no assumption about the distribution of colors in the adversarial clusters,
the matching of two vertices in an adversarial cluster can have a change in population of $O(1)$.  

Consider a random pair of vertices at the start of the evaluation phase.
For $m < \omha \N$, the expected change in population resulting from matching this pair of vertices is 
$\Omega(\N^{-1/2}) - o(\N^{-1/4}\N^{-1/4}) - o(\N^{-1/2}\cdot 1) = \Omega(\N^{-1/2})$.
Since the number of matched pairs of vertices in the evaluation phase is $\gamma m' = \Theta(\N)$,
it follows from linearity of expectation that the expected change in population 
during the evaluation phase is $\Omega(\sqrt{\N})$. 
Similarly, for  $m > \opha \N$ we have that the expected change in population
during the evaluation phase is $-\Omega(\sqrt{\N})$. 
The adversary can delete or insert only $\K \cdot \T = o(N^{1/4})$ agents during the epoch,
and Lemma~\ref{lem:correctround} implies that $O(N^{1/4})$ agents will self-destruct during procedure
$\Call{\CheckRoundConsistency}$ during the epoch, so the other terms dominate,
and the conclusion follows.
\end{proof}

\subsection{Putting everything together}
\label{ssec:finish-the-proof}

We now show that if the population size leaves the interval $[\omha\N, \opha\N]$
during an epoch, with high probability it will return to this interval during one of the next few epochs.
With this final lemma, we then conclude the proof of the theorem.

\begin{lemma}
\label{lem:combine}
Consider an epoch in which the population has 
drifted outside the interval $[\omha\N, \opha\N]$. With all but negligible probability,
the population will once again be in the interval $[\omha\N, \opha\N]$ at the start of one of the next
$t=\N^{0.01}$ epochs. 
\end{lemma}
\begin{proof}
Assume to the contrary, and let epoch 0 denote the first epoch after the population
has exceeded the interval $[\omha\N, \opha\N]$.
For concreteness, suppose the population has dropped below $\omha\N$. 
By Lemma \ref{lem:bounded_deviation}, with all but negligible probability
the population is still above
$\omha\N - \tilde{O}(\sqrt{\N})$.
For $i\in\{1,\ldots,t\}$, let $X_i$ be the random variable denoting
the difference in population between the start of epoch $i$ and the start
of epoch $i-1$, and let $\overline{X} =  (X_1+\cdots+X_t)/t$.
By Lemma \ref{lem:bounded_deviation}, with all but negligible probability,
each random variable $X_i$ is bounded in the range $[-\tilde{O}(\sqrt{\N}), \tilde{O}(\sqrt{\N})]$. 
Since by assumption the population is below $\omha\N$ at the start of each epoch,
Lemma \ref{lem:correct_in_exp} implies that $\mathbb{E}[X_i] = \Omega(\sqrt{\N})$. 
It follows by a Chernoff-Hoeffding bound that with all but negligible probability, for any constant $c$,
$\left|\overline{X} - \mathbb{E}[\overline{X}]\right| \leq c\sqrt{\N}$, so with all
but negligible probability we have that $\overline{X} = \Omega{\sqrt{\N}}$
and $X_1+\cdots+X_t = \Omega(\N^{0.51})$.
Consequently the population at the start of epoch $t$ exceeds $\omha\N$.
The argument is identical when the population has exceeded $\opha\N$. 
\end{proof}

We now conclude the proof of Theorem \ref{thm:mainthm-restated}. 
\begin{proof}[Proof of Theorem \ref{thm:mainthm-restated}]
Consider any polynomial $f$, and suppose that for some adversarial strategy
the population deviates from the interval $[(1-\alpha)\N, (1+\alpha)\N]$ 
in $f(\N)$ rounds with non-negligible probability. It follows that for some pair of epochs 
$i<j\in \left[1+\lfloor f(\N)/\T \rfloor \right]$,
with non-negligible probability the population deviates from interval $[(1-\alpha)\N, (1+\alpha)\N]$
for the first time in epoch $j$
after deviating from the interval $[\omha\N, \opha\N]$ for the last time in epoch $i$. We condition on this event.
Lemma \ref{lem:bounded_deviation} implies that until epoch $j$,
with all but neligible probability
the population will deviate in each epoch by at most $\tilde{O}(\sqrt{\N})$.
But then Lemma \ref{lem:combine} implies that with all but negligible 
probability the population will return to interval $[\omha\N, \opha\N]$
within $\N^{0.01}$ epochs, which is a contradiction.
Consequently the population will remain between $(1-\alpha)\N$ and $(1+\alpha)\N$
with high probability for any polynomial number of rounds, as desired.
\end{proof}

\bibliography{biblio.bib}

\newcommand{\etalchar}[1]{$^{#1}$}
\begin{thebibliography}{DMPTH10}

\bibitem[AAD{\etalchar{+}}04]{AADFP04}
Dana Angluin, James Aspnes, Zo{\"{e}} Diamadi, Michael~J. Fischer, and
  Ren{\'{e}} Peralta.
\newblock Computation in networks of passively mobile finite-state sensors.
\newblock In {\em Proceedings of the Twenty-Third Annual {ACM} Symposium on
  Principles of Distributed Computing, {PODC} 2004, St. John's, Newfoundland,
  Canada, July 25-28, 2004}, pages 290--299, 2004.

\bibitem[AAD{\etalchar{+}}06]{AADFP06}
Dana Angluin, James Aspnes, Zo{\"{e}} Diamadi, Michael~J. Fischer, and
  Ren{\'{e}} Peralta.
\newblock Computation in networks of passively mobile finite-state sensors.
\newblock {\em Distributed Computing}, 18(4):235--253, 2006.

\bibitem[AAE07]{AAE07}
Dana Angluin, James Aspnes, and David Eisenstat.
\newblock A simple population protocol for fast robust approximate majority.
\newblock In {\em Distributed Computing, 21st International Symposium, {DISC}
  2007, Lemesos, Cyprus, September 24-26, 2007, Proceedings}, pages 20--32,
  2007.

\bibitem[AAE{\etalchar{+}}17]{AAEGR17}
Dan Alistarh, James Aspnes, David Eisenstat, Rati Gelashvili, and Ronald~L.
  Rivest.
\newblock Time-space trade-offs in population protocols.
\newblock In {\em Proceedings of the Twenty-Eighth Annual {ACM-SIAM} Symposium
  on Discrete Algorithms, {SODA} 2017, Barcelona, Spain, Hotel Porta Fira,
  January 16-19}, pages 2560--2579, 2017.

\bibitem[AAER07]{AAER07}
Dana Angluin, James Aspnes, David Eisenstat, and Eric Ruppert.
\newblock The computational power of population protocols.
\newblock {\em Distributed Computing}, 20(4):279--304, 2007.

\bibitem[AAFJ08]{AAFJ08}
Dana Angluin, James Aspnes, Michael~J. Fischer, and Hong Jiang.
\newblock Self-stabilizing population protocols.
\newblock {\em {TAAS}}, 3(4):13:1--13:28, 2008.

\bibitem[AAG18]{AAG18}
Dan Alistarh, James Aspnes, and Rati Gelashvili.
\newblock Space-optimal majority in population protocols.
\newblock In {\em Proceedings of the Twenty-Ninth Annual ACM-SIAM Symposium on
  Discrete Algorithms}, pages 2221--2239. SIAM, 2018.

\bibitem[ABBS16]{ABBS16}
James Aspnes, Joffroy Beauquier, Janna Burman, and Devan Sohier.
\newblock Time and space optimal counting in population protocols.
\newblock In {\em 20th International Conference on Principles of Distributed
  Systems, {OPODIS} 2016, December 13-16, 2016, Madrid, Spain}, pages
  13:1--13:17, 2016.

\bibitem[BCG04]{BCG}
E.~R Berlekamp, John~Horton Conway, and R.~K. Guy.
\newblock Winning ways for your mathematical plays.
\newblock In {\em A K Peters 1st Ed.}, 2001-2004.

\bibitem[CDLN14]{CDLN14}
Alejandro Cornejo, Anna~R. Dornhaus, Nancy~A. Lynch, and Radhika Nagpal.
\newblock Task allocation in ant colonies.
\newblock In {\em Distributed Computing - 28th International Symposium, {DISC}
  2014, Austin, TX, USA, October 12-15, 2014. Proceedings}, pages 46--60, 2014.

\bibitem[Coo04]{Cook}
Matthew Cook.
\newblock Universality in elementary cellular automata, 2004.

\bibitem[DFGR06]{DFGR06}
Carole Delporte{-}Gallet, Hugues Fauconnier, Rachid Guerraoui, and Eric
  Ruppert.
\newblock When birds die: Making population protocols fault-tolerant.
\newblock In {\em Distributed Computing in Sensor Systems, Second {IEEE}
  International Conference, {DCOSS} 2006, San Francisco, CA, USA, June 18-20,
  2006, Proceedings}, pages 51--66, 2006.

\bibitem[DH97]{DH97}
Shlomi Dolev and Ted Herman.
\newblock Superstabilizing protocols for dynamic distributed systems.
\newblock {\em Chicago J. Theor. Comput. Sci.}, 1997.

\bibitem[Dij74]{Dj74}
Edsger~W. Dijkstra.
\newblock Self-stabilizing systems in spite of distributed control.
\newblock {\em Commun. {ACM}}, pages 643--644, 1974.

\bibitem[DLECW92]{EMCW92gnrh}
G~Martinez De~La~Escalera, AL~Choi, and Richard~I Weiner.
\newblock Generation and synchronization of gonadotropin-releasing hormone
  (gnrh) pulses: intrinsic properties of the gt1-1 gnrh neuronal cell line.
\newblock {\em Proceedings of the National Academy of Sciences},
  89(5):1852--1855, 1992.

\bibitem[DMPTH10]{DMTH10bacteria}
Tal Danino, Octavio Mondrag{\'o}n-Palomino, Lev Tsimring, and Jeff Hasty.
\newblock A synchronized quorum of genetic clocks.
\newblock {\em Nature}, 463(7279):326, 2010.

\bibitem[Gar70]{GameofLife}
Martin Gardner.
\newblock The fantastic combinations of john conway's new solitaire game
  ``life'', 1970.

\bibitem[GMRL15]{GMRL15}
Mohsen Ghaffari, Cameron Musco, Tsvetomira Radeva, and Nancy~A. Lynch.
\newblock Distributed house-hunting in ant colonies.
\newblock In {\em Proceedings of the 2015 {ACM} Symposium on Principles of
  Distributed Computing, {PODC} 2015, Donostia-San Sebasti{\'{a}}n, Spain, July
  21 - 23, 2015}, pages 57--66, 2015.

\bibitem[GR17]{GR17}
Oded Goldreich and Dana Ron.
\newblock On learning and testing dynamic environments.
\newblock {\em Journal of the ACM (JACM)}, 64(3):21, 2017.

\bibitem[JMPT87]{JMT87cardiology}
HJ~Jongsma, M~Masson-Pevet, and L~Tsjernina.
\newblock The development of beat-rate synchronization of rat myocyte pairs in
  cell culture.
\newblock {\em Basic research in cardiology}, 82(5):454--464, 1987.

\bibitem[LBBC14a]{LBBC14b}
Giuseppe Antonio~Di Luna, Roberto Baldoni, Silvia Bonomi, and Ioannis
  Chatzigiannakis.
\newblock Conscious and unconscious counting on anonymous dynamic networks.
\newblock In {\em Distributed Computing and Networking - 15th International
  Conference, {ICDCN} 2014, Coimbatore, India, January 4-7, 2014. Proceedings},
  pages 257--271, 2014.

\bibitem[LBBC14b]{LBBC14}
Giuseppe Antonio~Di Luna, Roberto Baldoni, Silvia Bonomi, and Ioannis
  Chatzigiannakis.
\newblock Counting in anonymous dynamic networks under worst-case adversary.
\newblock In {\em {IEEE} 34th International Conference on Distributed Computing
  Systems, {ICDCS} 2014, Madrid, Spain, June 30 - July 3, 2014}, pages
  338--347, 2014.

\bibitem[MO09]{MO'S09molecularclock}
Joseph~S Markson and Erin~K O'Shea.
\newblock The molecular clockwork of a protein-based circadian oscillator.
\newblock {\em FEBS letters}, 583(24):3938--3947, 2009.

\bibitem[Mor78]{Mor78}
Robert Morris.
\newblock Counting large numbers of events in small registers.
\newblock In {\em ommunications of the ACM}, pages 840--842, 1978.

\bibitem[vN51]{Neu}
John von Neumann.
\newblock The general and logical theory of automata, 1951.

\end{thebibliography}

\end{document}